\documentclass[11pt, draft]{article}

\usepackage{amsmath, amsfonts,amssymb, amsthm, euscript,makeidx,color,mathrsfs,latexsym,xcolor}

\newtheorem{assumption}{Assumption}

\newcommand{\ba}{\begin{array}}
\newcommand{\ea}{\end{array}}
\newcommand{\be}{\begin{equation}}
\newcommand{\ee}{\end{equation}}
\newcommand{\bee}{\begin{equation*}}
\newcommand{\eee}{\end{equation*}}
\newcommand{\bea}{\begin{eqnarray}}
\newcommand{\eea}{\end{eqnarray}}
\newcommand{\beaa}{\begin{eqnarray*}}
\newcommand{\eeaa}{\end{eqnarray*}}

\def\dbE{\mathbb{E}}
\def\dbF{\mathbb{F}}

\def\dbL{\mathbb{L}}

\def\dbP{\mathbb{P}}
\def\dbR{\mathbb{R}}

\def\bx{{\bf x}}

\def\a{\alpha}
\def\b{\beta}
\def\g{\gamma}

\def\e{\varepsilon}

\def\k{\kappa}
\def\l{\lambda}

\def\si{\sigma}

\def\th{\theta}

\def\D{\Delta}

\def\O{\Omega}
\def\cA{{\cal A}}
\def\cB{{\cal B}}

\def\cD{{\cal D}}

\def\cF{{\cal F}}

\def\cY{{\cal Y}}
\def\cZ{{\cal Z}}
\def\no{\noindent}

\def\ms{\medskip}
\def\bs{\bigskip}
\def\q{\quad}
\def\qq{\qquad}

\def\pa{\partial}
\def\cd{\cdot}
\def\cds{\cdots}

\def\td{\nabla}

\def\bx{{\bf x}}

\def\ol{\overline}
\def\ul{\underline}

\newcommand{\basa}{\begin{assumption}}
\newcommand{\easa}{\end{assumption}}

\newcommand{\bas}{\begin{assum}}
\newcommand{\eas}{\end{assum}}

\def\pa{\partial}

 \def\cd{\cdot}
\def\cds{\cdots}

\def\dis{\displaystyle}

\def\bx{{\bf x}}

\def\1{{\bf 1}}

\def\:{\!:\!}
\def\reff#1{{\rm(\ref{#1})}}
\def \proof{{\noindent \bf Proof\quad}}

at 9pt

\begin{document}

\newtheorem{thm}{Theorem}[section]
\newtheorem{lem}[thm]{Lemma}
\newtheorem{cor}[thm]{Corollary}
\newtheorem{prop}[thm]{Proposition}
\newtheorem{rem}[thm]{Remark}
\newtheorem{eg}[thm]{Example}
\newtheorem{defn}[thm]{Definition}
\newtheorem{assum}[thm]{Assumption}

\renewcommand {\theequation}{\arabic{section}.\arabic{equation}}
\def\thesection{\arabic{section}}

\title{\bf  A General Model for Continuous Time Principal-Agent Problem Under Hidden Action}

\author{
Jaeyoung Sung\thanks{ \no Ajou University, School of Business Administration, Suwon, 16499, South Korea.  Email: jaeyoungsung1@yahoo.com.}, ~ Jianfeng Zhang\thanks{\noindent
Department of Mathematics, University of Southern California, Los
Angeles, CA 90089. E-mail: jianfenz@usc.edu. This author is supported in part by NSF grant  DMS-2510403.
} ~ and ~{Zimu Zhu}\thanks{\noindent
Fintech Thrust, Hong Kong University of
Science and Technology (Guangzhou), Guangzhou, Guangdong Province, 511453, China. Email: zimuzhu@hkust-gz.edu.cn.
}}

\date{}
\maketitle

\begin{abstract}
In this paper, we study a general continuous-time Principal–Agent (PA) problem, where the agent privately makes effort and consumption decisions over time under a contract with payment schemes both in continuous time and in lump sums.  In particular, we allow the continuous payment process to be a controlled diffusion, which is directly related to the pay-to-performance sensitivity (PPS) in the empirical literature. In solving the agent’s problem, we propose a new sufficient condition that directly yields a solution to the agent’s problem, without requiring a separate verification step for the solution obtained from the first-order approach.  We also present an example which can be solved explicitly.
\end{abstract}

\vfill \bs

\no


\bs

\no{\it MSC 2020. 91B41, 91B43, 93E20.}  	

\eject

\section{Introduction}
In this paper, we study a general form of continuous time principal-agent problems with moral hazard. For exposition of general discrete time principal-agent problems, we refer to the books Bolton-Dewatripont \cite{BD-book}, Laffont-Martimort \cite{LM-book}, and Salanie \cite{Salanie-book}. In the seminal paper \cite{holmstrom1987aggregation}, Holmstrom and Milgrom (HM) introduce a continuous-time contracting problem between the agent and principal with their preferences characterized by exponential utility functions of aggregate effort and wealth, 
and strikingly show that their model yields the optimal contract in a closed form.  
In particular, the form turns out to be linear in the noisy terminal outcome of the agent's continuous effort.  
Sch\"{a}ttler and Sung \cite{schattler1993first} and Sung \cite{Sung} generalize the HM original model with general outcome processes by using the martingale method and dynamic programming in order to obtain the general structure of optimal contracts and their implementability conditions.  
For more general models, 
Cvitanic, Wan and Zhang \cite{cvitanic2009optimal} characterize the agent's optimal action and the principal's optimal contract via backward stochastic differential equations (BSDEs) and  forward-backward SDEs (FBSDEs), respectively. Cvitanić,  Possamaï and Touzi \cite{CPT2017, CPT2018} consider further models where the agent controls the volatility, 
and characterize the optimal controls/values via  second order BSDEs and HJB equations, respectively. 

The above works are mainly on lump sum payments. Many researchers have shown interests in modelling continuous-time contracting problems by allowing the agent to be continuously paid for his continuous effort.  Sannikov \cite{sannikov2008continuous} considers a model where the agent with felicity functions which are additively separable in effort and payment/consumtion is also concerned with his retirement utility which is triggered by the stochastic outcome of his effort.  The author derives a partial differential equation (PDE) for the optimal contract/the principal's value function, by utilizing the agent's remaining utility as a state variable.  Williams \cite{williams2015solvable} studies another continuous payment model where the agent is endowed with multiplicatively separable felicity functions, and the author successfully finds the optimal contract in a closed form.  Note that both authors allow their agents to be paid in Lebesgue (time-)integrals, but not in stochastic integrals.  As a consequence, their cumulative payments evolve with  drift but not with  diffusion.  We refer to the authors' books Cvitanic and Zhang \cite{CZbook} and Sung \cite{Sung-book} for general introduction of the subject in continuous time models.

In this paper, we present a continuous-time contracting model where the agent is paid not only continuously  during the contracting period, but also in lump sums at initial and terminal dates.  We believe that our model can help 
unify the two strands of continuous-time contracting literature on lump-sum and continuous payment models.
Moreover, those different pay schemes covered under the optimal contract are relevant to real-world applications.
The continuous payment can be viewed as a flow representation of salaries and bonuses during employment; and    
the initial and terminal lump sums as a signing bonus at hiring, and a golden parachute/severance pay at termination, respectively.  In particular, the initial payment of signing bonus seems new in the literature, to the best of our knowledge.

In sharp contrast with the existing continuous-payment models including \cite{sannikov2008continuous}  and \cite{williams2015solvable},  our model allows the continuous payment process to be fully stochastic with both drift and diffusion.  Our motivation for the diffusion-driven continuous payment process is not surprising, in light of the discrete-time contracting literature (such as Cheung \cite{cheung1969risksharing} and Stiglitz \cite{stiglitz1974incentives}), which suggests that the full stochasticity feature is necessary to improve both incentives and risk-sharing.  Hence we hope that our continuous-time model will help fill in the well-known gap in the continuous-payment literature.  In addition, we expect the same feature would enable us to relate our results directly to the pay-to-performance sensitivity (PPS) in the empirical literature.

In order to review briefly the motivation in the context of continuous-time contracting problems with continuous payments, let us recall that at each point in time during the whole contracting period, the agent's effort decision would be based on his concerns about the effort outcome which we assume is driven by an It\^{o} process with nontrivial drift and diffusion rates.   Then, one can easily imagine that the agent can be incentivized to work when his reward is properly tied to the stochastic effort outcome.  If it is notwithstanding restricted to be a drift process only, the payment/reward process may not be fully in tandem with the stochastic outcome/the It\^{o} process, failing to give the agent desired incentives.   Hence it is in general crucial to allow the payment process to be fully stochastic with both drift and diffusion, in providing optimal incentives for the agent to work for the principal.   The importance of the full stochasticity feature is even more obvious, in that it is also necessary for the efficient division of the fully stochastic outcome between the two individuals, particularly when they both are risk averse.


In our general model, the agent continuously and privately makes costly effort, and consumption decisions. 
As typically is the case in the literature, the agent's continuous consumption flow is supported by his savings account, indirectly affecting his terminal utility.  Constrained by the agent's private decisions, the principal chooses a contract affecting the agent's savings account balance.
The contract stipulates schemes in continuous and lump sum payments, which are allowed to be fully stochastic, in that their dynamics are typically characterized by It\^{o} processes with not only drifts but nontrivial diffusions.  

In the analysis of the agent's problem, the optimization over his efforts can be solved by the standard approach by using weak formulation and the comparison principle of BSDEs. However, his additional control on consumption requires optimization over controlled decoupled FBSDEs, which leads further to a coupled FBSDE as the first order condition. The well-posedness of the latter equation is in general very challenging. Indeed, in our general setting, neither the existence nor the uniqueness of the agent's optimal consumption is guaranteed.  By imposing some concavity conditions, for any given contract, we show that the agent's value function is concave in his consumption and thus solve the agent's problem completely.  In particular, this implies that the resulting FBSDE is also well-posed. Despite strong technical conditions, the last result seems new in the FBSDE literature.
  
Given the solution to the agent's problem and fixing any initial payment, by using the dynamic programming principle approach, we may characterize the principal's optimal contract via an HJB equation. The optimization over the initial payment is a finite dimensional optimization and can be solved easily. We shall note though, this HJB equation involves unbounded controls, and thus involves the face-lifting issue, namely the solution can be discontinuous at the terminal time. For the face-lifting, we refer to  Broadie, Cvitanic, and Soner \cite{BCS} and Soner and Touzi \cite{ST}, and also to Zhang and Zhu  \cite{ZZ} in the contexts of principal-agent problems.


These HJB equations are in general difficult to solve.  To obtain an explicit solution to the system, we consider a case with initial and continuous payments only, but without lump sum terminal payments. Both the principal and the agent are risk neutral in terms of the terminal wealth, but are risk averse in continuous payments. In such a model, we find that regardless of the risk-averse felicity functions, the PPS (the pay-to-performance sensitivity) of the optimal contract turns out to be equal to one. Moreover, if we restrict to absolutely continuous payment process as in the standard literature, we show that the principal's optimal value remains the same but the optimal contract does not exist. So the diffusion component of the continuous payment process also helps for the existence of optimal contracts. 



In summary, our main contributions are threefold. First, we provide a general model to the PA problem with both continuous payments and  initial and terminal lump-sum payments. One related work is  Williams \cite{williams2015solvable},  which simplify the model by applying the results of Cole and Kocherlakota \cite{cole2001efficient} when there is no PPS.  Second, in solving the agent’s problem, we propose a new sufficient condition that directly yields a solution, without requiring a separate verification step for the solution obtained through the first-order approach. This condition is new to the literature.  Moreover, as a by-product, we solve a class of FBSDEs that has not been covered in the existing literature. Third, in Section \S\ref{sect-example} we construct an example with continuous payments explicitly and obtain its explicit solution. In particular, we find the PPS helps for the existence of optimal contracts.




The rest of the paper is organized as follows. In \S\ref{sect-the model} we present the general model for continuous time principal-agent problem.  In \S\ref{sect-agent} we analyze the agent's optimal problem and characterize his optimal efforts and consumptions through a coupled system of FBSDEs. In \S\ref{sect-sufficient} we provide some sufficient conditions so as to solve the agent's problem rigorously. In \S\ref{sect-principal} we analyze the principal's problem and characterize its value through HJB equations. Finally in \S\ref{sect-example} we construct an example which admits an explicit solution.

%
%

\section{The General Model}
\label{sect-the model}
\setcounter{equation}{0}
In this section we introduce a general model of the principal agent problem on a finite time horizon $[0,T]$. Fix a complete probability space $(\O, \cF, \dbP)$,  a standard Brownian motion $B$, and we set the filtration $\dbF := \dbF^B$.  

The principal (she) hires an agent (he) to work for her. The outcome of the agent's efforts is the state process $X$. As usual we shall use the weak formulation, where the agent controls the distribution rather than the paths of $X$, so we  fix the process $X$ as follows:
\bea
\label{eq:X}
X_t = x_0 + \int_0^t \sigma(s,X_s) dB_s.
\eea
We assume $\si$ is Lipschitz continuous in $x$ and non-degenerate, thus the above SDE is well-posed and $\dbF^X=\dbF^B=\dbF$. 
The agent's control consists of two parts: an action process $a$ taking values in certain Euclidean space and a real valued consumption process $c^A$, both of which are $\dbF^X$-progressively measurable. Given an action process $a$, we introduce the probability measure $\dbP^a$ through the Girsanov theorem: for certain function $b$, 
\bea
\label{Girsanov}
\left.\ba{lll}
\dis B^a_t: = B_t - \int_0^t b^a_s ds,\q\mbox{where}\q b^a_s := \sigma^{-1}(s,X_s)b(s,X_s, a_s),\\
\dis d\dbP^a := M^a_T d\dbP\q\mbox{where}\q  M^a_t := \exp\Big(\int_0^t b^a_s dB_s - {1\over 2} \int_0^t |b^a_s|^2 ds\Big).
\ea\right.
\eea
Under appropriate technical conditions, which we will specify later, $B^a$ is an $\dbP^a$-Brownian motion, and we note that
\bea
\label{XPa}
X_t =  x_0 + \int_0^t b(s, X_s, a_s) ds + \int_0^t \si(s, X_s) dB^a_s,\q \dbP^a-\mbox{a.s.}
\eea
That is, $(\O, \cF, \dbP^a, B^a, X)$ is a weak solution to the above SDE. Throughout the paper,  the following notation shall be used:
\beaa
\dbE := \dbE^\dbP,\q \dbE^a := \dbE^{\dbP^a},\q  \dbE_t:= \dbE [\cd |\cF_t],\q \dbE_t^a:= \dbE^{a}[\cd |\cF_t].
\eeaa

 The principal's contract $\Xi$ consists of a constant $S_0\in \dbR$, a pair of $\dbR\times \dbR^d$-valued $\dbF^X$-progressively measurable processes $(\a, \b)$,  and an $\cF^X_T$-measurable $\dbR$-valued random variable $\xi$. Here $\xi$ stands for the lump sum payment, $S_0$ is the signing bonus, and $(\a, \b)$ compose the continuous payment $S$ in the form:
\bea
\label{S}
dS_t=  \alpha_t dt+ \beta_t dX_t.
\eea
We note that, by the martingale representation, one can express $\xi$ in terms of $\a, \b$. However, since the continuous payment and the lump sum payment will affect the agent's and the principal's utilities differently, our model is strictly more general than models considering only continuous payments.   Note also that the agent instantaneously receives for his current effort not only a fixed payment, $\alpha_t dt$, but a (short-term) performance-based payment $\beta_t dX_t$.  This feature is somewhat distinct from typical continuous payment models in the literature.   
We expect that our explicit modeling of the performance-based continuous payment, $\beta_t dX_t$, will help relate our results directly to the empirical literature.  
Moreover, one may conveniently think of the performance-based compensation embedded in $\xi$ as the agent's long-term incentives, even though it is unnecessary in our model to distinguish between the agent's long- and short-term performance incentives, because the agent's effort during the whole contracting period gives no impact on the principal's wealth after $T$.
The principal's control also includes her consumption rate $c^P$ (or dividend rate), which is a real valued and $\dbF^X$-progressively measurable process.  

Given the agent's control $(a, c^A)$ and the principal's control $(\Xi, c^P)=(S_0, \a, \b, \xi, c^P)$, their wealth processes evolve as follows: 
\bea
\label{wealth}
\left.\ba{c} 
\dis dm^A_t= (r_Am^A_t-c^A_t) dt+dS_t = [r_Am^A_t-c^A_t+\alpha_t] dt+\beta_t \sigma(t,X_t) dB_t,\\
\dis dm^P_t = (r_Pm^P_t-c^P_t)dt+dX_t-dS_t= \big[r_Pm^P_t-c^P_t-\alpha_t\big] dt+(1-\beta_t)\sigma (t,X_t)dB_t,\\
\dis m^A_0 = m^0_A+S_0,\q m^P_0 = m^0_P -S_0,
\ea\right.
\eea
where $r_A, r_P\ge 0$, $\si, m^0_A, m^0_P$ are given parameters, and their expected utilities are:
\bea
\label{value}
\left.\ba{lll}
\dis J_A (\Xi,a,c^A) :=\dbE^a\Big[\int_0^T f_A(s,X_s, \alpha_s,\beta_s,a_s,c^A_s)ds+g_A(T, X_T, \xi, m^A_T) \Big],\\
\dis J_P(\Xi, a, c^P) :=\dbE^a \Big[\int_0^T f_P(s,X_s, \alpha_s,\beta_s,  c^P_s)ds+g_P(T, X_T, \xi,m^P_T)\Big],
\ea\right.
\eea
where $(f_A, f_P)$ and $(g_A, g_P)$ are their  felicity and terminal utility functions, respectively. By standard BSDE theory, it is clear that
\bea
\label{J=Y}
 J_A (\Xi,a,c^A) = Y^{A,a}_0,\q J_P(\Xi, a, c^P) = Y^{P,a}_0,
 \eea
where, for $\k=A, P$,  $(Y^{\k,a}, Z^{\k, a})= (Y^{\k,\Xi ,a,c^\k}, Z^{\k, \Xi ,a,c^\k})$  solve the following BSDEs:
\bea
\label{BSDEa}
\left.\ba{lll}
\dis Y^{A,a}_t = g_A(T, X_T, \xi, m^A_T) + \int_t^T \big[f_A(s,X_s, \alpha_s,\beta_s,a_s,c^A_s) + Z^{A,a}_sb^a_s \big] ds - \int_t^T Z^{A,a}_s dB_s,\\
\dis Y^{P,a}_t  = g_P(T, X_T, \xi, m^P_T) +  \int_t^T \big[f_P(s,X_s, \alpha_s,\beta_s,  c^P_s) +  Z^{P,a}_sb^a_s \big] ds - \int_t^T Z^{P,a}_s dB_s.
\ea\right.
\eea
Here for notational simplicity, $b^a, B$ are column vectors, while $Z^{A,a}, Z^{P,a}$ are row vectors.

We now specify the technical conditions, so that all the above analysis are valid. 

\begin{assum}
\label{assum-0}
(i) All the coefficients $b, \si, f_A, g_A, f_P, g_P$ are progressively measurable in all variables. 


\no(ii) $\si$ is uniformly Lipschitz continuous in $x$ and is non-degenerate.

\no (iii) For any bounded $a$ and any $p\ge 1$, $\dbE[|M^a_T|^p]<\infty$.

\no(iv) For any bounded $\a, \b, a, c^A, c^P, \xi$, and any $S_0$, there exists $\e>0$ such that
\beaa
\dis\dbE\left[ 
\begin{aligned}
& \int_0^T \left[ |f_A(s,X_s, \alpha_s,\beta_s,a_s,c^A_s)|^{2+\e}+ |f_P(s, X_s, \alpha_s,\beta_s,  c^P_s)|^{2+\e}\right]ds \\
& \dis\qq + \left|g_A(T, X_T, \xi,m^A_T)\right|^{2+\e} + \left|\Phi^P(T, X_T, \xi, m^P_T)\right|^{2+\e} \end{aligned}\right]<\infty.
\eeaa
\end{assum}
To fulfill the technical requirements, we consider the admissible set of the agent's controls for effort and consumption $(a,c^A)$, $\cA_A = \cA_A^1 \times \cA_A^2$, and that of principal's controls for contracts $\Xi=(S_0,\alpha,\beta,\xi)$, and her own consumption $c^P$,  $\cA_P =  \cA_P^1 \times \cA_P^2$ where
\bea
\label{cA}
\left.\ba{l}
 \dis \cA_A^1  := \left\{ \, a \, \left| \, \begin{array}{l} \dbE[|M^a_T|^p] < \infty,~\mbox{for all}~ p\ge 1;~\mbox{and for any bounded $\a,\b, c^A$},\\
   \dbE\big[\int_0^T|f_A(s, X_s, \alpha_s,\beta_s,a_s,c^A_s)|^{2+\e}ds\big]<\infty, ~\mbox{for some}~ \e>0  \end{array}  \right. \right\},\\
 \dis \cA_A^2:= \{\mbox{bounded}~ c^A\},\q \cA^2_P :=\big\{\mbox{bounded}~ c^P\big\},\ms\\
 \dis \cA^1_P := \cA^{1,0}_P\times \cA^{1,1}_P\times \cA^{1,2}_P,\q \cA^{1,0}_P:= \dbR,~  \cA^{1,2}_P:= \big\{\mbox{bounded}~(\a,\b)\big\},\\ 
\dis \cA^{1,2}_P := \left\{ \, \xi \, \left| \, \begin{array}{l} \mbox{for any $S_0\in \cA^{1,0}_P$, $(\a,\b)\in \cA^{1,2}_P$, $c^A\in \cA^2_A$, $c^P\in \cA^2_P$, $\exists~\e>0$ }\\
\mbox{s.t.}~ \dbE\big[|g_A(T, X_T, \xi,m^A_T)|^{2+\e} + |\Phi^P(T,X_T, \xi, m^P_T)|^{2+\e}\big]<\infty \end{array} \right. \right\}.
\ea\right.
\eea
Note that 
 for any $(a, c^A)\in \cA_A$ and $(\Xi, c^P)\in \cA_P$, $\dbP^a$ is indeed a probability measure equivalent to $\dbP$, and BSDEs \reff{BSDEa} are wellposed. Moreover, by Assumption \ref{assum-0} and \reff{cA},
\beaa
\dbE^a \Big[\big|g_A(T, X_T, \xi,m^A_T)\big|^2\Big] = \dbE \Big[M^a_T \big|g_A(T, X_T, \xi,m^A_T)\big|^2\Big] <\infty,
\eeaa
and similarly,
\beaa
\dbE^a \Big[\int_0^T\!\!\! \big[\big|f_A(s,X_s, \alpha_s,\beta_s,a_s,c^A_s)\big|^2 + |f_P(s, X_s, \alpha_s,\beta_s,  c^P_s)|^2\big]ds +\big|g_P(T, X_T, \xi, m^P_T)\big|^2 \Big]  <\infty.
\eeaa

We are now ready to state the problems of the two contracting parities.

{\bf The Agent's problem.} Given a contract $\Xi\in\cA_P^1$,  the agent chooses $(a^*,c^{A*})\in \cA_A$ to maximize his utility,  which is the incentive compatibility (IC) condition:
\begin{align}
\label{VA}
J_A(\Xi,a^*,c^{A*})=\sup_{(a,c^A)\in \cA_A}  J_A(\Xi,a,c^A)=:V_A(\Xi), \qq (IC)
\end{align} 
We shall specify further conditions in the next section so that, for any $\Xi\in  \cA^1_P$, the agent's value $V_A(\Xi)<\infty$, and his optimal control exists (may not be unique).


{\bf The Principal's problem.} 
The principal chooses a contract $\Xi\in  \cA_P^1$ and a consumption rate $c^P\in\cA_P^2$ to maximize her utility:
\bea
\label{VP}
V_P=\sup_{(\Xi,c^P)\in \cA_P, (a^*, c^{A*})\in \cA_A} J_P(\Xi, a^*, c^P),
\eea
subject to the (IC) constraint and another the participant constraint (PC):
\bea
\label{PC condition}
V_A(\Xi)\geq R, \qq (PC) 
\eea
where $R\in \dbR$ is the reservation utility for the agent to participate in the game.  
Note that given an contract $\Xi$, the agent's optimal control  $(a^*, c^{A*})$ may not be unique.
Then, we assume, the agent chooses the one that is in the principal's best interest.

\section{The agent's problem}
\label{sect-agent}
\setcounter{equation}{0}
In this section we discuss general conditions which ensure the above stated agent's problem to be well-posed.  Then we will provide some verifiable sufficient conditions in the next section, and present an example in Section \ref{sect-example} which can be explicitly solved. 

Henceforth, we drop the sub/superscript $A$ whenever confusion is unlikely.  In particular, we write 
$$
c=c^A,\q m=m^A,\q (Y^a, Z^a)= (Y^{A,a}, Z^{A,a}).
$$ 
In this section we fix a contract $\Xi\in  \cA_P^1$, and  we assume zero interest rates:
\bea
\label{r=0}
r_A=r_P=0.
\eea  
This zero-interest assumption is without loss of generality, in that one can always reset the model in terms of $\tilde{m}^A_t$ and $\tilde{m}^P_t$, where  $\tilde{m}^A_t=e^{-r_At}m_t^A$ and $\tilde{m}^P_t=e^{-r_pt}m_t^P$.

We first optimize the agent's utility over $a$. In light of the first BSDE in \reff{BSDEa}, we introduce
\bea
\label{FA}
F_A(t, x, \a, \b, c, z) := \sup_a \big[f_A(t, x, \a, \b, a, c) + \si^{-1}(t,x) b(t,x,a)  z\big].
\eea

\begin{assum}
\label{assum-FA}
(i) $F_A<\infty$ is well-defined for any $(t,x,\a, \b, c, z)$, and for any $(\a, \b)\in \cA_P^{1,1}$ and $c\in \cA_A^2$, $\dbE\big[\int_0^T |F_A(t, X_t, \a_t, \b_t, c_t, 0)|^2dt\big]<\infty$;

\no (ii) For any  $(t,x,\a, \b, c, z)$, the Hamiltonian $F_A$ admits an optimizer $a^* = I_a(t,x,\a, \b, c, z)$;

\no(iii) $F_A$ is continuously differentiable in $z$; 

\no (iv) $\si^{-1} b$ is bounded, and thus $\pa_z F_A$ is bounded. 
\end{assum}
\no Note that under Assumption \ref{assum-FA}, we have (by the Envelope Thoorem)
\bea
\label{Ia*}
\si^{-1}(t,x) b\big(t,x,I_{a}(t,x,\a, \b, c, z)\big) = \pa_z F_A(t, x, \a, \b, c, z).
\eea
We next introduce the following BSDE, with $(\ol Y^{c}, \ol Z^{c}) = (\ol Y^{A,\Xi, c}, \ol Z^{A,\Xi, c})$, 
\bea
\label{BSDEc}
\dis \ol Y^{c}_t = g_A(T, X_T, \xi, m_T) + \int_t^T F_A\big(s,X_s, \alpha_s,\beta_s, c_s, \ol Z^{c}_s\big) ds - \int_t^T \ol Z^{c}_s dB_s.
\eea
In fact, this BSDE is resulted from that of the agent in (\ref{BSDEa}) when optimized with respect to $a$ given $(\Xi,c)$.   The following result is standard.
\begin{prop}
\label{prop-a*}
Let Assumptions \ref{assum-0} and \ref{assum-FA} hold. 

\no (i) For any $\Xi\in  \cA_P^1$ and $c\in \cA_A^2$,  BSDE \reff{BSDEc} is well-posed. 

\no (ii) $J_A (\Xi,a,c) \le \ol Y^{c}_0$ for all $a\in \cA_A^1$. Moreover, set $a^*_t := I_a(t,X_t, \alpha_t,\beta_t, c_t, \ol Z^{c}_t)$. If there exists $\e>0$ such that
\bea
\label{a*integrability}
\dbE\big[\int_0^T|f_A(s, X_s, \alpha_s,\beta_s,a^*_s,c_s)|^{2+\e}ds\big]<\infty,
\eea
then $a^*\in \cA_A^1$ is the agent's optimal action and $V_A (\Xi, c) = \ol Y^{c}_0$, where
\[
V_A(\Xi,c) := \sup_{a \in \cA_A^1} J_A(\Xi,a,c).
\]
\end{prop}
\proof Fix $\Xi\in \cA_P^1$ and $c\in \cA_A^2$. Since $F_A$ is uniformly Lipschitz continuous in $z$, and 
\beaa
\dbE\Big[|g_A(T, X_T, \xi, m_T)|^2 + \int_t^T |F_A\big(s,X_s, \alpha_s,\beta_s, c_s,0\big)|^2 ds\Big]<\infty,
\eeaa
it is clear that BSDE \reff{BSDEc} is well-posed. By the definition of $F_A$ and the comparison principle of BSDEs, we have $J_A (\Xi,a,c) \le \ol Y^{c}_0$ for all $a\in \cA_A^1$. Moreover, since $\si^{-1}b$ is bounded, we see that $\dbE[|M^{a^*}_T|^p]<\infty$ for all $p\ge 1$. Then \reff{a*integrability} implies that $a^*\in \cA_A^1$. It is obvious that $J_A (\Xi,a^*,c) = Y^{a^*}_0= \ol Y^{c}_0$. This implies that $V_A (\Xi, c) = \ol Y^{c}_0$ and $a^*$ is indeed an optimal action of the agent.
\qed

We next optimize the agent's utility $V_A (\Xi, c) = \ol Y^{c}_0$ over $c$. For this we shall apply the stochastic maximum principle, which works only for open loop controls. Here our controls are closed-loop, namely depending on the state process $X$. However, since we are using the weak formulation and thus $\dbF^X=\dbF^B$, this allows us to use the stochastic maximum principle. Recall \reff{wealth} and \reff{BSDEc}, and note that $m=m^{c}$ depends on $c$.  Fix $c, \D c\in \cA_A^2$ and set $c^{\e}:= c+\e \D c$. Introduce the variational processes, depending on $(c, \D c)$:
\bea
\label{tdY}
\td m := \lim_{\e\to 0} {1\over \e} [m^{c^{\e}} - m^{c}], \q \td \ol Y := \lim_{\e\to 0} {1\over \e} [\ol Y^{c^{\e}} - \ol Y^{c}],\q \td \ol Z := \lim_{\e\to 0} {1\over \e} [\ol Z^{c^{\e}} - \ol Z^{c}].
\eea
Under appropriate conditions and recalling \reff{r=0}, by straightforward calculation we obtain
\bea
\label{tdolY}
\left.\ba{c}
\dis  \td m = -\int_0^t  \D c_s ds,\\
\dis  \td\ol Y_t = \pa_m g_A\td m_T + \int_t^T \big[\pa_c F_A \D c_s + \pa_zF_A \td  \ol Z_s\big] ds - \int_t^T \td\ol Z_s dB_s.
\ea\right.
 \eea
 Introduce the conjugate processes $(\cY^{c}, \cZ^{c})$, which solve the following BSDE:
 \bea
 \label{cYA}
 \cY^{c}_t  = \pa_m g_A(T, X_T, \xi, m_T) + \int_t^T  \pa_z F_A\big(s,X_s, \alpha_s,\beta_s, c_s, \ol Z^{c}_s\big)\cZ^{c}_s ds -  \int_t^T \cZ^{c}_sdB_s.
 \eea
We note that this BSDE relies on $c$, but not on $\D c$. Then, for the $a^*$ in Proposition \ref{prop-a*} (ii), by  \reff{Ia*} and applying the It\^{o} formula we have
\beaa
d(\td \ol Y_t - \cY^c_t \td m_t) = -[\pa_c F_A - \cY^c_t] dt + [\td \ol Z_t - \cZ^c_t \td m_t] dB^{a^*}_t.
\eeaa
Since $\td\ol Y_0- \cY^c_0\td m_0 = \td \ol Y_0$, and $\td\ol Y_T- \cY^c_T\td m_T =0$, we obtain
\bea
\label{dual}
\td\ol Y_0 = \dbE^{a^*}\Big[\int_0^T \big[ \pa_c F_A\big(s,X_s, \alpha_s,\beta_s, c_s, \ol Z^{c}_s\big) - \cY^{c}_s\big] \D c_sds\Big].
\eea

We are now ready to state the necessary condition for the optimal $c^*$.

\begin{assum}
\label{assum-FA2}
(i) $F_A$ and $g_A$ are continuously differentiable in $c$ and $m$, respectively, such that, for any $\Xi\in \cA_P^1$ and $c\in \cA_A^2$, 
\beaa
\dbE\Big[\big| \pa_m g_A(T, X_T, \xi, m_T)|^2 + \int_t^T \big|\pa_c F_A\big(s,X_s, \alpha_s,\beta_s, c_s, \ol Z^{c}_s\big)\big|^2 ds\Big]<\infty.
\eeaa

\no(ii) $F_A$ is strictly concave in $c$; and the function $c\to y=\pa_c F_A\big(t, x, \alpha,\beta, c, \ol z\big)$ has an inverse function, denoted as $I_c(t,x,\a, \b, \ol z, y)$.
\end{assum}
\no We note that the concavity at (ii) above may not be ensured by the concavity of $f_A$ in $c$.

\begin{prop}
\label{prop-c*}
Let Assumptions \ref{assum-0}, \ref{assum-FA}, \ref{assum-FA2} hold, and assume \reff{r=0}. Assume $c^*\in \cA_A^2$ is optimal for the agent: $V_A(\Xi, c^*) = V_A(\Xi)$. Then $c^*_t = I_c(t, X_t, \a_t, \b_t, \ol Z_t, \cY_t)$, where $(m, \ol Y, \ol Z, \cY, \cZ) = (m^{\Xi}, \ol Y^{\Xi}, \ol Z^{\Xi}, \cY^{\Xi}, \cZ^{\Xi})$ satisfies the following FBSDE:
\bea
\label{FBSDE}
\left.\ba{lll}
\dis m_t = m_0 + \int_0^t \big[-I_c(s, X_s, \a_s, \b_s, \ol Z_s, \cY_s)+\alpha_s\big] ds+ \int_0^t \beta_s \sigma(s,X_s) dB_s;\\
\dis \ol Y_t = g_A(T, X_T, \xi, m_T)  + \int_t^T F^{I_c}_A\big(s,X_s, \alpha_s,\beta_s, \ol Z_s, \cY_s\big) ds- \int_t^T \ol Z_s dB_s,\\
 \dis \cY_t  = \pa_m g_A(T, X_T, \xi, m_T) + \int_t^T  (\pa_z F_A)^{I_c}\big(s,X_s, \alpha_s,\beta_s, \ol Z_s, \cY_s\big)\cZ_s ds -  \int_t^T \cZ_sdB_s,\\
 \mbox{where, for $\phi = F_A, \pa_z F_A$},~ \phi^{I_c}(t,x,\a, \b, \ol z, y) := \phi\big(t,x,\a, \b, I_c(t, x, \a, \b, \ol z, y), \ol z\big). 
\ea\right.
\eea
\end{prop}
\proof Under Assumption \ref{assum-FA2}, clearly the FBSDEs \reff{tdolY} and \reff{cYA} are well-posed, and one can easily verify that \reff{tdY} and \reff{dual} hold true. Given optimal $c^*$, by \reff{tdY} we have $\td \ol Y_0=\td \ol Y^{c^*,\D c}_0 \le 0$ for all $\D c\in \cA_A^2$, then \reff{dual} implies $ \pa_c F_A\big(s,X_s, \alpha_s,\beta_s, c^*_s, \ol Z^{c^*}_s\big) = \cY^{c^*}_s$, $ds\times d\dbP^{a^*}$-a.s., which is equivalent to $ds\times d\dbP$-a.s. Thus we have $c^*_s = I_c(s, X_s, \a_s, \b_s, \ol Z_s, \cY_s)$. Plugging this into \reff{wealth},  \reff{tdolY}, and \reff{cYA}, we obtain the FBSDE \reff{FBSDE}.
\qed

\begin{rem}
\label{rem-agent}
(i) The system of FBSDEs in \reff{FBSDE} is a necessary condition, rather than a sufficient condition for the optimal $c^*$. By using Ekeland's variational principle, we can show that if FBSDE \reff{FBSDE} is well-posed for all $\Xi$, including the continuous dependence on $\Xi$, then $c^*_t := I_c(t, X_t, \a_t, \b_t, \ol Z_t, \cY_t)$ is indeed optimal for the agent's problem, see \cite[Theorem 10.3.10]{CZbook}.

\no (ii) The well-posedness of FBSDE \reff{FBSDE} is in general quite challenging. In fact, without additional structural conditions, one may expect neither the existence nor the uniqueness of the optimal control for the optimization problem $\sup_{c\in \cA_A^2} V_A(\Xi, c)$. 

\no (iii) In the next section, we shall provide some sufficient conditions, so that  the optimization problem $\sup_{c\in \cA_A^2} V_A(\Xi, c)$ has a unique optimal $c^*$, and consequently the corresponding FBSDE \reff{FBSDE} is well-posed.
\end{rem}
%
%

\section{A sufficient condition for the agent's problem}
\label{sect-sufficient}
\setcounter{equation}{0}
In this section we assume $f_A$ takes the following form
\bea
\label{FA-special}
\left.\ba{lll}
\dis f_A(t, x, \a, \b, a, c) = f_1(t,x,\a, \b) + f_2(t,x,a) -  \l_1(t) e^{-\l_2(t) c},\\
\dis \mbox{thus}\q F_A(t,x,\a, \b, c, z) = f_1(t,x,\a, \b) + F_2(t,x,z) -  \l_1(t) e^{\l_2(t) c},\\
\dis\mbox{where}\q F_2(t,x,z) := \sup_a \big[f_2(t,x,a)  + \si^{-1}(t,x) b(t,x,a)  z\big].
\ea\right.
\eea

\begin{assum}
\label{assum-FA3}
Assume \reff{FA-special} holds true.

\no(i) $F_2$ is twice continuously differentiable in $z$ with $0\le \pa_{zz} F_2 \le L$ for a constant $L$.

\no (ii) For $i=1,2$, $\ul \l_i \le \l_i(t) \le \ol \l_i$, for some constants $\ol \l_i \ge \ul \l_i > 0$. 

\no (iii) $g_A$ is twice continuously differentiable in $m$, such that $\ul \g_T \le \pa_m g_A \le \ol \g_T$ for some constants $\ol \g_T \ge \ul \g_T >0$, and $\pa_{mm} g_A \le 0$.

\no (iv) $4LT|\ol \gamma_T|^2 < \underline{\lambda}_2 \ul \gamma_T$.
\end{assum}

In this case 
\bea
\label{Ic2}
I_c(t, y) = {1\over \l_2(t)} \ln\Big({\l_1(t)\l_2(t)\over y}\Big),\q y>0,
\eea
and FBSDE \reff{FBSDE} becomes
\bea
\label{FBSDE2}
\left.\ba{lll}
\dis m_t = m_0 + \int_0^t \big[- {1\over \l_2(s)} \ln\Big({\l_1(s)\l_2(s)\over \cY_s}\Big)+\alpha_s\big] ds+ \int_0^t \beta_s \sigma(s,X_s) dB_s;\\
\dis \ol Y_t = g_A(T, X_T, \xi, m_T)  - \int_t^T \ol Z_s dB_s\\
\dis\qq \q + \int_t^T \big[f_1(s,X_s,\a_s, \b_s) + F_2(s,X_s, \ol Z_s) - {\cY_s\over \l_2(s)}\big] ds,\\
 \dis \cY_t  = \pa_m g_A(T, X_T, \xi, m_T) + \int_t^T  \pa_z F_2\big(s,X_s,  \ol Z_s\big)\cZ_s ds -  \int_t^T \cZ_sdB_s.
\ea\right.
\eea

\begin{prop}
\label{prop-bound}
Let Assumptions \ref{assum-0}, \ref{assum-FA}, \ref{assum-FA2}, \ref{assum-FA3} (i)-(iii), and \reff{r=0} hold true.  Introduce
\bea
\label{K}
\ul K_t:= {1\over \lambda_2(t)}\ln\Big({\lambda_1(t)\lambda_2(t)\over \ol \g_T}\Big), \q \ol K_t:= {1\over \lambda_2(t)}\ln\Big( {\lambda_1(t)\lambda_2(t) \over \ul\g_T}\Big).
\eea
Clearly $\ul K\le \ol K$. Fix $\Xi\in \cA_P^1$ and denote $J(c) := V_A(\Xi, c) = \ol Y^c_0$. Then 
\bea
\label{cbound}
J(c) \le \min\big( J(c\vee \ul K),~ J(c\wedge \ol K) \big),\q \forall c\in \cA_A^2.
\eea
Consequently,
\bea
\label{ctruncation}
V(\Xi) = \sup_{c\in \cA_A^2([\ul K, \ol K])} J(c),\q\mbox{where}\q \cA_A^2([\ul K, \ol K]):=\Big\{c\in \cA_A^2: \ul K \le c\le \ol K\Big\}.
\eea
\end{prop} 
\proof  We shall only prove  $J(c)\leq J(c\wedge \ol K)$.  The other inequality follows from similar arguments. Fix $c\in \cA_A^2$ and denote  
 $$
 \Delta c: = c-c\wedge \ol K =(c- \ol K)^+\geq 0.
 $$
Denote $\Delta m_t := m_t^c-m_t^{c\wedge \ol K}$, $\D \ol Y :=\ol Y^c -\ol Y^{c\wedge \ol K}$, $\D\ol Z:= \ol Z^c - \ol Z^{c\wedge \ol K}$. Then 
\beaa
\dis \Delta m_t &=&  -\int_0^t \Delta c_s ds;\\
\dis \Delta \ol Y_t &=& \pa_m g_A(T,X_T, \xi, \tilde{m}_T) \Delta m_T  -\int_t^T \D\ol Z_sdB_s\\
\dis &&+\int_t^T\big[\pa_z F_2(s, X_s, \tilde Z_s) \D \ol Z_s +\l_1(s)\l_2(s) e^{-\l_2(s) \tilde c_s} \D c_s\big]ds\\
\dis &=&\int_t^T\big[\l_1(s)\l_2(s) e^{-\l_2(s) \tilde c_s} - \pa_m g_A(T,X_T, \xi, \tilde{m}_T)\big]\D c_sds - \int_t^T \D\ol Z_sd \tilde B_s,
\eeaa
where $\tilde{c}_s \in [c_s\wedge \ol K_s, c_s]$, $\tilde{m}_T\in [m_T^{c\wedge \ol K}, m_T^c]$, $\tilde Z_s$ is between $\ol Z^c$ and $\ol Z^{c\wedge \ol K}$,  and 
\beaa
\tilde B_t := B_t - \int_0^t \pa_z F_2(s, X_s, \tilde Z_s) ds,
\eeaa
is $\tilde \dbP$-Brownian motion, with $\tilde \dbP$ defined in an obvious way through the Girsanov theorem. Note that $\tilde c_s > \ol K_s$ whenever $\D c_s\neq 0$. Then, by Assumption \ref{assum-FA3} (iii) and \reff{K}, we have
\beaa
\dis \Delta \ol Y_0 &=& \dbE^{\tilde \dbP}\Big[\int_0^T\big[\l_1(s)\l_2(s) e^{-\l_2(s) \tilde c_s} - \pa_m g_A(T,X_T, \xi, \tilde{m}_T)\big]\D c_sds \Big] \\
\dis &\le&  \dbE^{\tilde \dbP}\Big[\int_0^T\big[\l_1(s)\l_2(s) e^{-\l_2(s) \ol K_s} - \ul \g_T\big]\D c_sds \Big] = 0,
\eeaa
with equality holding true if and only if $\D c =0$. This implies that $J(c)\le J(c\wedge \ol K)$, with equality holding true if and only if $c \le \ol K$.
\qed

\begin{prop}
\label{prop-concave}
Let Assumptions \ref{assum-0}, \ref{assum-FA}, \ref{assum-FA2}, \ref{assum-FA3}, and \reff{r=0} hold true.  Fix $\Xi\in \cA_P^1$ and let $J(c)$ be as in Proposition \ref{prop-bound}. Then $J$ is strictly concave on $\cA_A^2([\ul K, \ol K])$. 
\end{prop} 
\proof  Fix $c^0, c^1\in  \cA_A^2([\ul K, \ol K])$, and denote $\D c := c^1-c^0$, $c^\th = c^0  + \th \D c$, $\th \in [0,1]$, and $\tilde J(\th) := J(c^\th)$. By \reff{tdolY}, \reff{cYA}, and \reff{FA-special}, we have 
\bea
\label{1storder}
{d\over d\th} \tilde J(\th) = \td \ol Y^\th_0,
\eea
where, denoting $(m^\th, \ol Y^\th,\ol Z^\th) := (m^{c^\th}, \ol Y^{c^\th}, \ol Z^{c^\th})$, 
\bea
\label{tdYth}
\left.\ba{lll}
\dis  \td m_t = -\int_0^t  \D c_s ds,\\
\dis  \td\ol Y^\th_t = \pa_m g_A(T, X_T, \xi, m^{\th}_T)\td m_T - \int_t^T \td\ol Z^\th_s dB_s\\
 \dis\qq + \int_t^T \big[\l_1(s)\l_2(s)e^{-\l_2(s) c^\th_s} \D c_s + \pa_zF_2(s, X_s, \ol Z^{\th}_s) \td  \ol Z^\th_s\big] ds.
\ea\right.
 \eea
By differentiating in $\th$ further, we have
\beaa
{d^2\over d\th^2} \tilde J(\th) = \td^2 \ol Y^\th_0,
\eeaa
where
\bea
\label{tdYth2}
\left.\ba{lll}
\dis  \td^2\ol Y^\th_t = \pa_{mm} g_A(T, X_T, \xi, m^{\th}_T)|\td m_T|^2 - \int_t^T \td^2\ol Z^\th_s dB_s + \int_t^T \Big[\pa_zF_2(s, X_s, \ol Z^{\th}_s) \td^2  \ol Z^\th_s\\
 \dis\qq\qq -\l_1(s)|\l_2(s)|^2e^{-\l_2(s) c^\th_s} |\D c_s|^2 + (\td  \ol Z^\th_s)^\top \pa_{zz}F_2(s, X_s, \ol Z^{\th}_s) \td  \ol Z^\th_s\Big] ds.
\ea\right.
 \eea
Let $\dbP^\th$ denote the probability measure induced by $B^\th_t := B_t - \int_0^t \pa_zF_2(s, X_s, \ol Z^{\th}_s) ds$. Then
\beaa
&&\dis {d^2\over d\th^2} \tilde J(\th) = \td^2 \ol Y^\th_0 = \dbE^{\dbP^\th}\Big[\pa_{mm} g_A(T, X_T, \xi, m^{\th}_T)|\td m_T|^2\\
&&\dis\qq + \int_0^T \big[ -\l_1(s)|\l_2(s)|^2e^{-\l_2(s) c^\th_s} |\D c_s|^2 + (\td  \ol Z^\th_s)^\top \pa_{zz}F_2(s, X_s, \ol Z^{\th}_s) \td  \ol Z^\th_s\big] ds\Big].
\eeaa
Note that $\ul K \le c^\th_s \le \ol K$. By \reff{K} and Assumption \ref{assum-FA3} (ii) we have
\beaa
\l_1(s)|\l_2(s)|^2 e^{-\l_2(s) c^\th_s} \ge \l_1(s)|\l_2(s)|^2 e^{-\l_2(s) \ol K_s} = \l_2(s) \ul \g_T\ge \ul \l_2  \ul \g_T.
\eeaa
Then, by Assumption \ref{assum-FA3} (i) we have
\beaa
 {d^2\over d\th^2} \tilde J(\th) \le   \dbE^{\dbP^\th}\Big[\int_0^T \big[ L |\td  \ol Z^\th_s|^2-\ul \l_2\ul \g_T   |\D c_s|^2 \big]ds\Big].
\eeaa
On the other hand, by \reff{tdYth} we have
\beaa
&&\dis \dbE^{\dbP^\th}\Big[| \td\ol Y^\th_0|^2 + \int_0^T  |\td\ol Z^\th_t|^2dt\Big]\\
&&\dis = \dbE^{\dbP^\th}\Big[\Big(\pa_m g_A(T, X_T, \xi, m^{\th}_T)\td m_T + \int_0^T \l_1(s)\l_2(s)e^{-\l_2(s) c^\th_s} \D c_sds\Big)^2\Big]\\
&&\dis \le 2 \dbE^{\dbP^\th}\Big[|\pa_m g_A(T, X_T, \xi, m^{\th}_T) \td m_T|^2 + T\int_0^T |\l_1(s)\l_2(s)e^{-\l_2(s) c^\th_s} \D c_s|^2ds\Big]\\
&&\dis \le 2 \dbE^{\dbP^\th}\Big[\ol\g_T^2 T \int_0^T |\D c_s|^2ds  + T\int_0^T |\l_1(s)\l_2(s)e^{-\l_2(s) \ul K_s} |^2  |\D c_s|^2ds\Big]\\
&&\dis =2 \dbE^{\dbP^\th}\Big[\ol\g_T^2 T \int_0^T |\D c_s|^2ds  + T\int_0^T \ol\g_T^2  |\D c_s|^2ds\Big] = 4\ol\g_T^2 T \dbE^{\dbP^\th}\Big[\int_0^T |\D c_s|^2ds\Big] .
\eeaa 
Then, by Assumption \ref{assum-FA3} (iv),
\beaa
\dis {d^2\over d\th^2} \tilde J(\th)  \le \dbE^{\dbP^\th}\Big[\int_0^T \big[4 L\ol\g_T^2 T -\ul \l_2\ul \g_T \big] |\D c_s|^2ds\Big]\le 0,
\eeaa
with equality holding true if and only if $\D c =0$. That is, for any different $c^0, c^1\in  \cA_A^2([\ul K, \ol K])$, $\tilde J$ is strictly concave on $[0, 1]$. Therefore, for any $\th\in (0,1)$, we have
\beaa
J(\th c^1 + (1-\th) c^0) = J(c^\th) = \tilde J(\th) < \th \tilde J(1) + (1-\th) \tilde J(0) = \th J(c^1) + (1-\th) J(c^0).
\eeaa
This means that $J$ is strictly concave on $\cA_A^2([\ul K, \ol K])$. 
\qed

\begin{thm}
\label{thm-special}
Consider the setting  in Proposition \ref{prop-concave}. Then $J$ has a unique optimal control $c^*$ in $\cA_A^2$ and  FBSDE \reff{FBSDE2} has a unique solution. Moreover, $c^*_t = I_c(t, \cY_t)$ as in \reff{Ic2} and $c^*\in \cA_A^2([\ul K, \ol K])$. 
\end{thm} 
\proof We proceed in three steps.

{\bf Step 1.} We first prove the existence of $c^*$. Recall \reff{ctruncation} and let $c^n\in \cA_A^2([\ul K, \ol K])$, $n\ge 1$, be an approximating optimal sequence, namely $\lim_{n\to\infty} J(c^n) = V_A(\Xi)$. Note that $\ul K$, $\ol K$ are bounded, then $\{c^n\}_{n\ge 1}$ is weakly compact. That is,  by otherwise considering a subsequence, there exists $c^*\in \cA_A^2([\ul K, \ol K])$ such that 
\beaa
\lim_{n\to\infty} \dbE\Big[\int_0^T c^n_t \eta_t dt\Big] = \dbE\Big[\int_0^T c^*_t \eta_t dt\Big],\q\forall \eta\in \dbL^2(\dbF).
\eeaa
Moreover, by Mazur's lemma (cf. \cite{Rudin}), there exist subsequence $\hat c^m$ converging to $c^*$ strongly:
\bea
\label{convstrong}
 \lim_{m\to\infty} \dbE\Big[\int_0^T |\hat c^m_t- c^*_t|^2  dt\Big] = 0,
\eea
where each $\hat c^m$ is a convex combination of $\{c^n\}_{n\ge m}$. That is, for each $m$, there exist $n^m_1, \cds, n^m_{k_m} \ge m$ and $\l^m_1, \cds, \l^m_{k_m}$ such that $\hat c^m = \sum_{j=1}^{k_m} \l^m_j c^{n^m_j}$ and $\l^m_j\ge 0$, $\sum_{j=1}^{k_m} \l^m_j = 1$. By Proposition \ref{prop-concave}, we have
\beaa
J(\hat c^m) \ge \sum_{j=1}^{k_m} \l^m_j J(c^{n^m_j}) \to V_A(\a,\b, \xi),\q\mbox{as}~ m\to\infty.
\eeaa
By \reff{convstrong}, it is clear that $\lim_{m\to\infty} J(\hat c^m) = J(c^*)$. Then $J(c^*)\ge V_A(\Xi)$. Since $V_A(\Xi)$ is the optimal value, then we must have the equality and hence $c^*$ is an optimal control.

{\bf Step 2.} We next prove the uniqueness of $c^*$. Assume $c^*\in \cA_A^2$ is an arbitrary optimal control. First, by \reff{cbound} we see that $c^*\vee \ul K$ and $c^*\wedge \ol K$ are also optimal, and $J(c^*) = J(c^*\vee \ul K) = J(c^*\wedge \ol K)$. By the arguments in the end of the proof for Proposition \ref{prop-bound}, we see that $c^* = c^*\vee \ul K = c^*\wedge \ol K$. That is, $\ul K\le c^* \le \ol K$, and thus $c^* \in \cA_A^2([\ul K, \ol K])$. Now by the strict concavity of $J$ on $\cA_A^2([\ul K, \ol K])$, as obtained in Proposition \ref{prop-concave}, we see that $c^*$ is unique.

{\bf Step 3.} By Proposition \ref{prop-c*}, $c^*$ induces a solution to FBSDE \reff{FBSDE2}, and it holds that  $c^*_t = I_c(t,\cY_t)$. In particular, since $\ul \g_T \le \pa_m g_A\le \ol \g_T$ and $\pa_z F_2=\pa_z F_A$ is bounded, by the third equation in \reff{FBSDE2} we see that $\ul \g_T \le \cY_t\le \ol \g_T$, $\dbE[\int_0^T |\cZ_t|^2dt\big]<\infty$.  Then by the first two equations in \reff{FBSDE2} we have $\dbE\big[\sup_{0\le t\le T}[|m_t|^2 + |\ol Y_t|^2] + \int_0^T |\ol Z_t|^2dt\big]<\infty$.

It remains to prove the uniqueness of FBSDE \reff{FBSDE2}.  Let $m$, $(\ol Y, \ol Z)$, $(\cY, \cZ)$ be an arbitrary solution to FBSDE \reff{FBSDE2}. Denote $c_t := I_c(t, \cY_t)$. Recall \reff{tdY}, \reff{tdolY}, \reff{cYA}, and \reff{dual}, and let $\tilde \dbP$ denote the probability measure induced by  $\tilde B_t := B_t-\int_0^t \pa_z F_2\big(s,X_s,\ol Z^{c}_s\big)ds$ through Girsanov theorem.
By the derivation of FBSDE \reff{FBSDE} or \reff{FBSDE2} we have $\td \ol Y^{c}_0 = 0$ for any $\D c$. In particular, setting $\D c_t := (c_t - \ol K_t)^+\ge 0$, by \reff{dual}, \reff{cYA}, and \reff{FA-special} we have
\beaa
0 &=& \td \ol Y^{c}_0 =  \dbE^{\tilde \dbP}\Big[\int_0^T \big[ \l_1(s)\l_2(s) e^{-\l_2(s) c_s}  - \cY^{c}_s\big] \D c_sds\Big]\\
&=&\dbE^{\tilde \dbP}\Big[\int_0^T \big[ \l_1(s)\l_2(s) e^{-\l_2(s) c_s}  - \pa_m g_A(T, X_T, \xi, m^c_T)\big] \D c_sds\Big]\\
&\ge& \dbE^{\tilde \dbP}\Big[\int_0^T \big[ \l_1(s)\l_2(s) e^{-\l_2(s) c_s}  - \ul \g_T\big] \D c_sds\Big]\\
&=& \dbE^{\tilde \dbP}\Big[\int_0^T \l_1(s)\l_2(s) \big[e^{-\l_2(s) c_s}  - e^{-\l_2(s) \ol K_s}\big] \D c_sds\Big]
\eeaa
Note that $e^{-\l_2(s) c_s}  - e^{-\l_2(s) \ol K_s}>0$ whenever $\D c_s>0$. Then we must have $\D c_s = 0$, $ds\times d\tilde \dbP$-a.s., which is equivalent to $ds\times d\dbP$-a.s. This implies that $c\le \ol K$. Similarly, we have $c\ge \ul K$. Then $c\in \cA_A^2([\ul K, \ol K])$.

We now assume FBSDE \reff{FBSDE2} has two different solutions $m^i$, $(\ol Y^i, \ol Z^i)$, $(\cY^i, \cZ^i)$, $i=0,1$, and set $c^i_t := I_c(t, \cY^i_t)$. Then $c^0, c^1 \in  \cA_A^2([\ul K, \ol K])$, $\td \ol Y^{c^i}_0=0$ for any $\D c$, and $c^0, c^1$ are not equal. Denote $\D c := c^1-c^0$, $c^\th := c^0 + \th \D c$, $\tilde J(\th) := J(c^\th)$, as in Proposition \ref{prop-concave}. By \reff{1storder} we have ${d\over d\th} \tilde J(0) = \td  \ol Y^{c^0}_0=0$, ${d\over d\th} \tilde J(1) = \td  \ol Y^{c^1}_0=0$. This contradicts with Proposition \ref{prop-concave} that $\tilde J$ is strictly concave. Therefore, FBSDE \reff{FBSDE2} has at most one solution.
\qed

\begin{rem}
\label{rem-FBSDE}
In Theorem \ref{thm-special}, we use the control problem to prove the well-posedness of the corresponding FBSDE \reff{FBSDE2}, rather than the opposite direction as in the standard literature. To the best of our knowledge, albeit under strong technical conditions, the well-posedness of FBSDE \reff{FBSDE2}  is new in the literature.  
\end{rem}

\section{The Principal's problem}
\label{sect-principal}
\setcounter{equation}{0}
We now turn to the principal's problem. The problem is overall very challenging, and in fact we do not expect its well-posedness in such a general form.  Namely the principal's problem may not have an optimal contract, as the agent's problem given a contract may not have an optimal control, and even if they exist, one may not expect in general their uniqueness. Nevertheless, in this section we provide some partial results here, and then we will present a solvable example rigorously and completely in the next section. 

As in the literature, we rewrite the BSDEs in \reff{FBSDE} in the forward form, and consider general initial time $t_0$. That is, given $t_0<T$, together with the equations for $X$ and $m^P$ and recalling \reff{r=0}, we consider the following forward equations on $[t_0, T]$:
\bea
\label{forward}
\left.\ba{lll}
\dis X_t= x_0+\int_{t_0}^t \sigma(s,X_s)dB_s,\\
\dis m^A_t=m_A +\int_{t_0}^t [\alpha_s-I_c(s, X_s, \alpha_s,\beta_s, \ol Z_s, \cY_s)]ds +\int_{t_0}^t \beta_s \sigma(s,X_s) dB_s,\\
\dis m^P_t= m_P+\int_{t_0}^t [-\alpha_s-c^P_s]ds +\int_0^t(1-\beta_s)  \sigma(s,X_s)  dB_s,\\
\dis \ol Y_t = \ol y -\int_{t_0}^t  F^{I_c}_A\big(s,X_s, \alpha_s,\beta_s, \ol Z_s, \cY_s\big)  ds+\int_{t_0}^t \ol Z_sdB_s, \\
\dis \cY_t= y' -\int_{t_0}^t (\pa_z F_A)^{I_c}\big(s,X_s, \alpha_s,\beta_s, \ol Z_s, \cY_s\big)\cZ_s ds +\int_{t_0}^t \cZ_s dB_s,
\ea\right.
\eea
with the following terminal constraints:
\bea
\label{terminal-constraint}
\ol Y_T = g_A(T, X_T, \xi, m^A_T),\q \cY_T  = \pa_m g_A(T, X_T, \xi, m^A_T),\q\mbox{for some $\cF_T$-measurable}~ \xi.
\eea
Here $(\a, \b)\in \cA_P^{1,1}$.  Let 
\begin{align*}
\cA_Z   &:= 
    \left\{ \ol Z, \cZ \in \dbL^2(\dbF) \left| \begin{array}{l}
     \mbox{$\bullet$ The SDEs \reff{forward} are well-posed;} \\
     \mbox{$\bullet$  The process $(\pa_z F_A)^{I_c}\big(s,X_s, \alpha_s,\beta_s, \ol Z_s, \cY_s\big)$ induce a probability } \\ 
     \mbox{ $\quad$ measure $\dbP^{\ol Z, \cZ}$ through the Girsanov theorem.}
    \end{array} \right. \right\}
\end{align*}
To circumvent the constraint on $\xi$ which is not convenient to handle, we focus on two cases.

\begin{assum}
\label{assum-lumpsum}
One of the following two cases holds true:

\no{\bf Case 1.} There is no lump sum payment, namely $\xi$ is not involved. Then \reff{terminal-constraint} becomes
\bea
\label{terminal-constraint1}
\ol Y_T = g_A(T, X_T, m^A_T),\q \cY_T  = \pa_m g_A(T, X_T,  m^A_T).
\eea

\no{\bf Case 2.} $\pa_\xi g_A>0$.  Let $g_A^{-1}(T, x, \ol y, m)$ denote the inverse function of $g_A$ with respect to $\xi$.  Then we may reduce \reff{terminal-constraint} into one constraint:
\bea
\label{terminal-constraint2}
\cY_T  = \pa_m g_A\Big(T, X_T, g_A^{-1}\big(T, X_T, \ol Y_T, m^A_T\big), m^A_T\Big).
\eea
\end{assum}

Denote
\bea
\label{hatgP}
\hat g_P(T, X_T, m^P_T, m^A_T, \ol Y_T) := \left\{ \ba{lll} g_P(T, X_T, m^P_T),\q\mbox{in Case 1};\\
g_P\big(T, X_T, g_A^{-1}(T, X_T, \ol Y_T, m^A_T), m^P_T\big),\q \mbox{in Case 2}.
\ea\right.
\eea
We now introduce the dynamic value function:
\bea
\label{Vtx}
V(t_0, x, m_A, m_P, \ol y, y') := \sup_{(\a, \b,  c^P)\in \cA^{1,1}_P\times \cA_P^3, (\ol Z, \cZ)\in \cA_Z} Y^P_{t_0},\q\mbox{subject to \reff{terminal-constraint1} or  \reff{terminal-constraint2}},
\eea
where, recalling \reff{Ia*},
\bea
\label{BSDEP}
\left.\ba{lll}
\dis Y^P_t  = \hat g_P(T, X_T, m^P_T, m^A_T, \ol Y_T)  - \int_t^T Z^{P}_s dB_s\\
\dis\q +  \int_t^T \Big[f_P(s,X_s, \alpha_s,\beta_s,  c^P_s) +  (\pa_z F_A)^{I_c}(s, X_s, \a_s, \b_s, \ol Z_s, \cY_s) Z^{P}_s \Big] ds,~ t\in [t_0, T].
 \ea\right.
 \eea
 Equivalently,
 \bea
 \label{BSDEP2}
\left.\ba{lll}
\dis Y^P_{t_0}  = \dbE^{\dbP^{\ol Z, \cZ}}_{t_0}\Big[\hat g_P(T, X_T, m^P_T, m^A_T, \ol Y_T)  +  \int_{t_0}^T f_P(s,X_s, \alpha_s,\beta_s,  c^P_s)  ds\Big].
 \ea\right.
 \eea
 We note that, although $Y^P_{t_0}$ is $\cF_{t_0}$-measurable, as in the standard control theory, its essential supremum is deterministic. By abusing the notation we write supremum  in \reff{Vtx} and $V$ is a deterministic function. 
 We emphasize that at this point we do not require FBSDE \reff{FBSDE} to be sufficient for the agent's optimal control, and thus the value function $V(0,\cd)$ may not solve the agent's problem in general. We shall establish this connection in the end of this section, and at this point we focus on the function $V$, which has independent interest. 
 
 We also note that, given $(t_0, x, m_A, m_P, \ol y, y')$, there could be no admissible controls $(\a, \b, \ol Z, \cZ)$ satisfying the terminal constraint ($c^P$ is not involved in the terminal constraint). In this case, we set $V(t_0, x, m_A, m_P, \ol y, y') := -\infty$. To ensure the value function is well defined, we impose the following assumption. For any $t_0<T$, let $D_{t_0}$ denote the set of  $\bx=(x, m_A, m_P,\ol y, y')$ such that there exist $(\a,\b)\in \cA_P^{1,1}$ and $(\ol Z, \cZ) \in \cA_Z$ satisfying the terminal constraint \reff{terminal-constraint1} or  \reff{terminal-constraint2}.  Denote further that
 \bea
 \label{cD}
 \cD := \big\{(t_0, \bx): t_0<T, \bx\in D_{t_0}\big\},
 \eea
 and let $\cD^\circ$, $\ol \cD$, and $\pa \cD$ denote its interior, closure, and boundary, respectively. Recall the initial value $m^0_A, m^0_P$ in \reff{wealth}.
 \begin{assum}
\label{assum-Dt0}
 (i) There exists $\bx\in D_0$ such that $m_A + m_P = m^0_A+m^0_P$;  
 
\no (ii)  $V<\infty$ on $\cD$ and $V$ is continuous  in $\cD^\circ$. 
 \end{assum}

We next study the properties of $V$. When there is no terminal constraint, the optimization problem \reff{Vtx} naturally leads to the following HJB equation:
\bea
\label{HJB}
\pa_t V(t,x,m_A, m_P, \ol y, y') + \sup_{\a, \b, c_P, \ol z, z'\in \dbR} \dbL V\big(t,x,m_A, m_P, \ol y, y'; \a, \b, c_P, \ol z, z'\big) =0,
\eea
where
\bea
\label{dbLV}
\left.\ba{lll}
\dis \dbL V := {1\over 2} \Big[\pa_{xx} V \si^2 + \pa_{m_Am_A} V \b^2 \si^2  + \pa_{m_Pm_P} V (1-\b)^2\si^2 + \pa_{\ol y\ol y} V \ol z^2 + \pa_{y'y'} z'^2\Big]  \\
\dis \q + \pa_{x m_A}V \b \si^2 + \pa_{xm_P}V (1-\b) \si^2 + \pa_{x\ol y}V \si\ol  z + \pa_{xy'} V \si z' + \pa_{m_Am_P} \b (1-\b) \si^2 \\
\dis\q + \pa_{m_A \ol y} V \b \si \ol z + \pa_{m_A y'} V \b \si z' + \pa_{m_P \ol y} V (1-\b) \si \ol z + \pa_{m_P y'} (1-\b)\si z' + \pa_{\ol y y'} V \ol zz' \\
\dis\q + \pa_{m_A} V \big[\a - I_c(t, x, \alpha,\beta, \ol z, y')\big] - \pa_{m_P} V [\a + c] - \pa_{\ol y} V F^{I_c}_A\big(t,x, \alpha,\beta, \ol z, y'\big)\\
\dis\q - \pa_{y'} V (\pa_{z} F_A)^{I_c}\big(t,x, \alpha,\beta, \ol z, y'\big)  z' + f_P(t,x, \alpha,\beta,  c_P) \\
\dis \q+ \big[\pa_x V \si + \pa_{m_A} V \b \si + \pa_{m_P} V (1-\b)\si + \pa_{\ol y} V \ol z + \pa_{y'} V z'\big]  (\pa_{z} F_A)^{I_c}\big(t,x, \alpha,\beta,\ol  z, y'\big).
\ea\right.
\eea

\begin{prop}
\label{prop-DPP}
Let Assumptions \ref{assum-0}, \ref{assum-FA}, \ref{assum-FA2}, \ref{assum-lumpsum}, \ref{assum-Dt0}, and \reff{r=0} hold true. Then, for any $t_0< t_1 < T$ and any $\bx\in D_{t_0}$, we have
\bea
\label{DPP}
V(t_0, \bx) \le \sup_{(\a, \b, c^P)\in \cA_P^{1,1}\times \cA_P^3, (\ol Z, \cZ)\in \cA_Z} Y^{P, t_1, V(t_1,\cd)}_{t_0},
\eea
where, for the processes in \reff{forward},
\bea
\label{BSDEP3}
\dis Y^{P, t_1, V(t_1,\cd)}_{t_0}  := \dbE^{\dbP^{\ol Z, \cZ}}_{t_0}\Big[V(t_1, X_{t_1}, m^A_{t_1}, m^P_{t_1}, \ol Y_{t_1}, \cY_{t_1}) +  \int_t^{t_1} f_P(s,X_s, \alpha_s,\beta_s,  c^P_s)  ds\Big].
 \eea
 Consequently, $V$ is a viscosity subsolution of HJB equation \reff{HJB} in $\cD^\circ$.
\end{prop}
\proof Given $\bx\in D_{t_0}$, for any $(\a,\b)\in \cA_P^{1,1}$ and $(\ol Z, \cZ) \in \cA_Z$ satisfying the terminal constraint, it is clear that $(X_{t_1}, m^A_{t_1}, m^P_{t_1}, \ol Y_{t_1}, \cY_{t_1}) \in D_{t_1}$, a.s. Then
\beaa
Y^P_{t_0} &=& \dbE^{\dbP^{\ol Z, \cZ}}_{t_0}\left[\hat g_P(T, X_T, m^P_T, m^A_T, \ol Y_T)  +  \int_{t_0}^T f_P(s,X_s, \alpha_s,\beta_s,  c^P_s)  ds\right]\\
&=&  \dbE^{\dbP^{\ol Z, \cZ}}_{t_0}\left[ \begin{array}{l} \dbE^{\dbP^{\ol Z, \cZ}}_{t_1}\left[\hat g_P(T, X_T, m^P_T, m^A_T, \ol Y_T)  +  \int_{t_1}^T f_P(s,X_s, \alpha_s,\beta_s,  c^P_s)  ds\right]\\
\quad  +\int_{t_0}^{t_1} f_P(s,X_s, \alpha_s,\beta_s,  c^P_s)  ds \end{array} \right]\\
&\le& \dbE^{\dbP^{\ol Z, \cZ}}_{t_0}\left[V(t_1, X_{t_1}, m^A_{t_1}, m^P_{t_1}, \ol Y_{t_1}, \cY_{t_1}) +  \int_t^{t_1} f_P(s,X_s, \alpha_s,\beta_s,  c^P_s)  ds\right] = Y^{P, t_1, V(t_1,\cd)}_{t_0}.
\eeaa
Take supremum over $(\a,\b,c^P, \ol Z, \cZ)$ on both sides, we obtain \reff{DPP}.

We note that the partial DPP \reff{DPP} does not involve any constraint. Then the viscosity subsolution property of $V$ follows from rather standard arguments.
\qed

\begin{rem}
\label{rem-DPP}
(i) We expect the full DPP holds true, namely equality holds in \reff{DPP}, and then $V$ is a viscosity solution of HJB \reff{HJB}. However, due to the terminal constraint, the opposite direction is much harder to prove. We refer to \cite{ST, ZZ} for some positive results along this direction, and to \cite{LXZ} for a numerical algorithm for HJB equations arising from principal-agent problems. 

\no (ii) The full DPP, as well as the continuity of $V$ assumed in Assumption \ref{assum-Dt0},  are much more likely to hold if we allow the admissible controls to satisfy the constraints only approximately. This amounts to the subtle difference between the raw set value and the set value in \cite{FRZ}. We shall investigate the problem under such relaxed constraints in future study.
\end{rem}

We finally discuss the connection between the value function $V$ and the principal's optimal value $V_P$ in \reff{VP}, subject to the constraints (IC) and (PC). We first note that the (PC) constraint \reff{PC condition} amounts to $\ol y \ge R$, and recall again the initial value $m^0_A, m^0_P$ in \reff{wealth}. 
We thus introduce:
\bea
\label{DP}
D_P:= \big\{(S_0, \ol y, y')\in \dbR^3:  \ol y \ge R, (x_0, m^0_A+S_0, m^0_P-S_0, \ol y, y')\in D_0\big\}.
\eea 
Next, note that 
\beaa
V(T,x,m_A, m_P, y, y') =  \hat g_P(T, x, m_P, m_A, y).
\eeaa
However, since the controls $\ol z, z'$ in HJB \reff{HJB} are unbounded, in general there is a face-lifting issue. That is, $V$ may not be left continuous at $t=T$, see e.g. \cite{BCS, ST, ZZ}. Indeed, assume the following limit exists:
\bea
\label{facelifting}
V(T-,x,m_A, m_P, \ol y, y') := \lim_{t\uparrow T} V(t,x,m_A, m_P,\ol y, y')
\eea
In general  $V(T-,\cd) \neq V(T,\cd)$. More generally, recalling \reff{cD},  we define the boundary value of $V$ and assume its existence:
\bea
\label{facelifting2}
V_{-}(t,\bx) := \lim_{(t',\bx')\in \cD^\circ \to (t, \bx)} V(t',\bx'),\q (t,\bx)\in \pa \cD.
\eea

\begin{thm}
\label{thm-principal}
Let all the conditions in Proposition \ref{prop-DPP} hold true. Assume further that

\no (i) The boundary value function $V_{-}$ in \reff{facelifting2} exists on $\pa\cD$, and $V_{-}\ge V$ on $\cD \cap \pa\cD$;

\no (ii) The equation \reff{HJB} on $\ol \cD$ has a viscosity solution $U$ with boundary condition $V_{-}$ and comparison principle of viscosity semisolutions holds; or $U$ is a classical solution and the partial comparison principle between a classical solution and a viscosity subsolution holds;

\no (iii) For any $(t_0, \bx)\in \cD$ and $\e>0$, there exist $(\a, \b, \xi, c^P)\in \cA^{1,1}_P\times \cA^{1,2}_P\times \cA^2_P$ and $(\ol Z, \cZ)\in \cA_Z$ satisfying the constraint \reff{terminal-constraint} such that $U(t_0, \bx)  \le \dbE[Y^P_{t_0}]+\e$, for the $Y^P_{t_0}$ in \reff{BSDEP} or  \reff{BSDEP2}.

\no (iv) $D_P \neq \emptyset$, and for any $(S_0, \ol y, y')\in D_P$ and $\e>0$, there exist  $(\a,\b, \xi, c^P, \ol Z, \cZ)$ satisfying (iii) above for $(t_0, \bx) = (0, x_0, m^0_A+S_0, m^0_P-S_0, \ol y, y')$ and $\ol y = V_A(\a,\b, \xi)$.

 Then $U = V$ in $\cD$, and, recalling \reff{VP},
\bea
 \label{Principal=supV}
 V_P = \sup_{(S_0, \ol y, y')\in D_P} U(0, x_0, m^0_A+S_0, m^0_P-S_0, \ol y, y').
 \eea
\end{thm}
\proof First, by Proposition \ref{prop-DPP} and the comparison principle we have $V \le U$ on $\cD^\circ$. Next, by the condition (iii) we see that $U \le V$ in $\cD$.  These, together with (i) that $U = V_{-}\ge V$ on $\cD\cap \pa\cD$, imply that $U=V$ in $\cD$.

It remains to prove \reff{Principal=supV}. Let $\tilde V_P$ denote the right side of \reff{Principal=supV}. First, let $\Xi\in \cA_P$ and $(a^*, c^{A*})\in \cA_A$ satisfy the constraints (IC) and (PC).  Let $(\ol Z, \cZ)$ be determined by \reff{FBSDE}. By Propositions \ref{prop-a*}, \ref{prop-c*} and the (IC) constraint we see that $(S_0, \a, \b, c^P)\in \cA_P^{1,0}\times \cA_P^{1,1}\times \cA_P^2$ and $(\ol Z, \cZ)\in \cA_Z$ satisfy constraint \reff{terminal-constraint1} or  \reff{terminal-constraint2}. Moreover, by (PC) constraint we have $\ol y = \ol Y_0 = V_A(\Xi)\ge R$. That is, recalling \reff{wealth}, $(S_0, \ol y, \cY_0)\in D_P$. Then, for the $Y^P$ in \reff{BSDEP2}, we have $J_P(\Xi,a^*, c^P)\le Y^P_0 \le \tilde V_P$. Thus $V_P\le \tilde V_P$. 

On the other hand, for any $(S_0, \ol y, y')\in D_P$ and denoting $\bx :=  (x_0, m^0_A+S_0, m^0_P-S_0, \ol y, y')$ , by condition (iv) above, in particular since $\ol y = V(\Xi)$,  we see that 
\beaa
Y^P_0 = J_P(\Xi, a^*, c^P) \le V_P.
\eeaa
  Then by (iii), $U(0, \bx)\le Y^P_0 +\e \le  V_P+\e$. Since $\e>0$ and $(S_0, \ol y, y')\in D_P$ are arbitrary, we obtain $\tilde V_P\le V_P$, and thus equality holds.
\qed

\begin{rem}
\label{rem-principal1}
(i) Given the constraint \reff{terminal-constraint}, it is not clear that if $V$ is decreasing in $\ol y$ in general, thus we cannot set $\ol y=R$ in \reff{Principal=supV}.

\no(ii) When there is face-lifting, especially when $V(T-,\cd) \neq V(T,\cd)$,  in general the optimal contract may not exist.
\end{rem}

We shall remark that the conditions in Theorem \ref{thm-principal} are rather strong, and some of them are hard to verify. Nevertheless, we will solve an example completely in the next section.

\section{A solvable  example}
\label{sect-example}
\setcounter{equation}{0}

 In this section, we present a simple example which has a closed form solution. We remark that it may violate some conditions in our assumptions in previous sections. However, since it can be solved explicitly, all the involved well-posedness are guaranteed and the results are rigorously correct.
 
Consider the following model with continuous payments only:
 \bea
 \label{model}
 \left.\ba{c}
 b(t,x,a)=a,\q \sigma(t,x)=1,\q r_A = r_P=0,\q m^0_A=m^0_P=0,\\
 f_A(t,x, \alpha,\beta, a,c^A)= f(c^A- {1\over 2} |a|^2),\q f_P(t,x, \alpha,\beta, c^P)= c^P - \k(t) (1-\b)^2,\\
 g_A(T, x, m_A)=m_A,\q g_P(T, x, m_P)=m_P.
 \ea\right.
 \eea

\begin{assum}
\label{assum-example}
(i) $f$ is continuously twice differentiable and strictly concave, with $\lim_{x\to -\infty} f'(x) =\infty$, $\lim_{x\to \infty} f'(x) =-\infty$, where $f'$ denotes the derivative of $f$.  Let $h:=(f')^{-1}$ denote the inverse function of $f'$. 

\no(ii) $\k\ge 0$ and is bounded.
\end{assum}

\subsection{The agent's problem}
We first solve the agent's problem. The main trick here is that $g_A$ is linear in $m_A$, which allows us to optimize both $a$ and $c$ through the comparison principle of BSDEs. To be precise, recall \reff{BSDEa} and \reff{wealth}, for $(m^{A, c^A}, Y^{A,a, c^A}, Z^{A,a, c^A})=(m^{A}, Y^{A,a}, Z^{A,a})$ we have 
\beaa
\left.\ba{lll}
&&\dis Y^{A,a, c^A}_t = m^{A,c^A}_T + \int_t^T \big[f(c^A_s-{1\over 2} a_s^2) + a_s Z^{A,a, c^A}_s \big] ds - \int_t^T Z^{A,a, c^A}_s dB_s\\
&&\dis\qq  =m^{A,c^A}_t + \int_t^T \big[\a_s - c^A_s + f(c^A_s-{1\over 2} a_s^2) + a_s Z^{A,a, c^A}_s \big]ds - \int_t^T [Z^{A,a, c^A}_s-\b_s] dB_s.
\ea\right.
\eeaa
Denote 
\beaa
\tilde Y^{A,a, c^A}_t := Y^{A,a, c^A}_t - m^{c^A}_t,\q \tilde Z^{A,a, c^A}_t := Z^{A,a, c^A}_t - \b_t.
\eeaa
Then
\beaa
\tilde Y^{A,a, c^A}_t = \int_t^T \big[\a_s - c^A_s + f(c^A_s-{1\over 2} a_s^2) + a_s (\tilde Z^{A,a, c^A}_s + \b_s) \big]ds - \int_t^T \tilde Z^{A,a, c^A}_s dB_s.
\eeaa
Introduce
\bea
\label{FA2}
\tilde F(z):= \sup_{a,c} \big[-c + f(c-{1\over 2}a^2) + a z\big].
\eea
Note that, by changing the variable $\tilde c := c-{1\over 2} a^2$, we have
\bea
\label{FA3}
\left.\ba{lll}
\dis \tilde F(t, \tilde z) = \sup_{a, \tilde c} \big[-[\tilde c+{1\over 2} a^2] + f(\tilde c) + a (\tilde z+\b_t)\big] \\
\dis\qq\q = {1\over 2} (\tilde z+\b_t)^2 + f(h(1)) - h(1) = {1\over 2} (\tilde z+\b_t)^2 + \tilde F(0),\\
\dis \mbox{with}\q a^*_t=I_a(t, \tilde z) := \tilde z+\b_t,\q \tilde c^* = h(1),\q  c^*_t=I_c(t,\tilde z) := {1\over 2} (\tilde z+\b_t)^2+ h(1).
\ea\right.
\eea
Now given $(\a, \b)\in \cA_P^{1,1}$, consider the BSDE:
\bea
\label{tildeBSDE}
\tilde Y^A_t = \int_t^T \big[\a_s + {1\over 2}|\tilde Z^{A}_s + \b_s|^2 + \tilde F(0) \big]ds - \int_t^T \tilde Z^{A}_s dB_s.
\eea
Note that the above BSDE has quadratic growth in $\tilde Z^A$.  Let $\cB$ denote the set of process $Z$ such that $\int_0^t Z_s dB_s$ is a BMO martingale. Moreover, note that $m^A_0=S_0$ and
\beaa
\sup_{a, c^A} Y^{A,a,c^A}_0 = \sup_{a,c^A}  \tilde Y^{A,a,c^A}_0 + S_0.
\eeaa
Then, by standard BSDE theory we have
\begin{prop}
\label{prop-eg-agent}
Under \reff{model} and Assumption \ref{assum-example}, for any $(\a, \b)\in \cA_P^{1,1}$, the BSDE \reff{tildeBSDE} is wellposed with $ \tilde Z^A\in \cB$, and
\bea
\label{eg-agent}
V_A(\Xi) = \tilde Y^A_0 + S_0,~ \mbox{with optimal controls}~ a^*_t = \tilde Z^A_t+\b_t,~ c^{A*}_t = {1\over 2} |\tilde Z^A_t+\b_t|^2 + h(1).
\eea
\end{prop}
\no We note that, since in \reff{model} $g_A, g_P$ do not depend on $x$, $V_A(\Xi)$ here does not depend on $\xi$.
 
 \subsection{The dynamic value function}
 We next investigate the dynamic value function $V$ corresponding to \reff{Vtx}. First we rewrite \reff{tildeBSDE} in the forward form. Given $(t_0,x)$ and $(\Xi, c^P)\in \cA_P$, $\tilde y\in \dbR$, $\tilde Z\in \cB$, consider the following forward equations on $[t_0, T]$ (again they do not depend on $\xi$):
 \bea
\label{forward2}
\left.\ba{lll}
\dis X_s= x+ B_t-B_{t_0},\\
\dis m^A_t=S_0 +\int_{t_0}^t \big[\alpha_s- {1\over 2} |Z_s|^2 - h(1)\big]ds +\int_0^t \beta_s  dB_s,\\
\dis m^P_t= -S_0+\int_{t_0}^t [-\alpha_s-c^P_s]ds +\int_0^t(1-\beta_s)   dB_s,\\
\dis \tilde Y_t = \tilde y -\int_{t_0}^t  \big[\a_s + {1\over 2}|\tilde Z_s + \b_s|^2 + \tilde F(0) \big]  ds+\int_{t_0}^t \tilde Z_sdB_s.
\ea\right.
\eea 
Noting that $b^{a^*} = a^* = \tilde Z+\b$, we consider the BSDE:
\bea
\label{YP-eg}
Y^{P}_t  = m^P_T +  \int_t^T \big[c^P_s - k(s) (1-\b_s)^2  +  (\tilde Z_s+\b_s) Z^{P}_s\big] ds - \int_t^T Z^{P}_s dB_s.
\eea
We note that this BSDE is not Lipschitz in $Z^P$. However, since $\tilde Z\in \cB$, by using Girsanov theorem one can easily prove its well-posedness.  We then define
\bea
\label{Vtx2}
V(t_0,x, S_0,\tilde y) := \sup_{\a,\b,c^P, \tilde Z} \dbE[Y^P_{t_0}],\q\mbox{subject to}~\tilde Y_T=0.
\eea

We observe that this problem can be simplified, due to its special structure. Indeed, we first note that $X$ and $m^A$ are not involved in \reff{YP-eg}. Next, denote
\beaa
\tilde Y^P_t := Y^P_t - m^P_t,\q \tilde Z^P_t := Z^P_t - (1-\b_t).
\eeaa
Then
\bea
\label{YP-eg2}
\left.\ba{lll}
\dis \tilde Y_t = \tilde y -\int_{t_0}^t  \big[\a_s +{1\over 2}|\tilde Z_s + \b_s|^2 + \tilde F(0) \big]  ds+\int_{t_0}^t \tilde Z_sdB_s,\\
\dis \tilde Y^{P}_t  = \int_t^T \big[-\a_s - k(s) (1-\b_s)^2  +  (\tilde Z_s+\b_s) (\tilde Z^{P}_s + 1-\b_s)\big]  ds - \int_t^T \tilde Z^{P}_s dB_s.
\ea\right.
\eea 
This implies that
\bea
\label{Vtx3}
\left.\ba{c}
\dis V(t_0,x, S_0, \tilde y) = U(t_0, \tilde y) -S_0,\ms\\
\dis \mbox{where}\q U(t_0, \tilde y):= \sup_{\a,\b, \tilde Z} \dbE[\tilde Y^P_{t_0}],\q\mbox{subject to}~\tilde Y_T=0.
 \ea\right.
\eea

 We now study the function $U$.  First, note that for any $t_0<T$ and $\tilde y\in \dbR$, set 
 \bea
\label{eg-control}
\a_t := {\tilde y\over T-t_0} - C_0,\q \b_t :=1, \q\tilde Z_t:=0,\q t\in [t_0, T],\q \mbox{where}\q C_0:= {1\over 2} + \tilde F(0).
\eea
That is, $D_{t_0} = \dbR$ for all $t_0<T$. 
 \begin{thm}
 \label{thm-eg-U}
 Let \reff{model} and Assumption \ref{assum-example} hold true. Then
\bea
\label{eg-classical}
U(t, \tilde y) = C_0 (T-t) - \tilde y,\q t<T, \tilde y\in \dbR.
\eea
\end{thm}
\proof First, given $t_0<T$ and $\tilde y$, consider the controls in \reff{eg-control}, we have
\beaa
\tilde Y_T=0,\q \tilde Y^P_{t_0} = -\int_{t_0}^T \a_t dt = - \tilde y + C_0 (T-t).
\eeaa
Then $U(t_0, \tilde y) \ge  C_0 (T-t_0) - \tilde y$.

On the other hand, denote $\hat Y^P_t := \tilde Y^P_t + \tilde Y_t$ and $\hat Z^P_t := \tilde Z^P_t + \tilde Z_t$.  For any $\a, \b, \tilde Z$ satisfying the constraint $\tilde Y_T =0$, we have
\beaa
&&\dis\hat Y^P_{t_0}  = \int_{t_0}^T \big[{1\over 2}|\tilde Z_s+\b_s|^2 + \tilde F(0)  - k(s) (1-\b_s)^2 \\
&&\dis \qq\qq +  (\tilde Z_s+\b_s) (\hat Z^{P}_s- \tilde Z_s + 1-\b_s)\big]  ds - \int_{t_0}^T \hat Z^{P}_s dB_s\\
&&\dis =\int_{t_0}^T \big[-{1\over 2}|\tilde Z_s+\b_s-1|^2 +C_0  - k(s) (1-\b_s)^2  +  (\tilde Z_s+\b_s) \hat Z^{P}_s\big]  ds - \int_{t_0}^T \hat Z^{P}_s dB_s\\
&&\dis \le   C_0(T-t_0) + \int_{t_0}^T   (\tilde Z_s+\b_s) \hat Z^{P}_s  ds - \int_{t_0}^T \hat Z^{P}_s dB_s.
\eeaa
Since $\tilde Z\in \cB$ and $\b$ is bounded, one can easily see that
\beaa
\tilde Y^P_{t_0} + \tilde y = \hat Y^P_{t_0} \le C_0(T-t_0).
\eeaa
Since $\a,\b, \tilde Z$ are arbitrary, we obtain $U(t_0, \tilde y) + \tilde y  \le C_0(T-t_0)$, proving \reff{eg-classical}.
 \qed

\begin{rem}
\label{rem-eg-facelift}
(i) Note that $\lim_{t_0\uparrow T} U(t_0, \tilde y) = - \tilde y$, while from \reff{YP-eg2} and \reff{Vtx3} we have $U(T, \tilde y) = 0$. So there exists facelifting issue here.

\no(ii) When $\k = 0$, one may easily check that the model \reff{model} is equivalent to the following model which has only lump sum payment $\xi$, without continuous payment $(\a, \b)$:  
 \bea
 \label{model2}
 \left.\ba{c}
 b(t,x,a)=a,\q \sigma(t,x)=1,\q r_A = r_P=0, \q m^0-A=m^0_P=0,\q \a=\b=0,\\
 f_A(t,x, \alpha,\beta, a,c^A)= f(c^A- {1\over 2} |a|^2),\q f_P(t,x, \alpha,\beta, c^P)= c^P,\\
 g_A(T, x, \xi, m_A)= m_A+\xi,\q g_P(T, x, \xi, m_P)=m_P - \xi.
 \ea\right.
 \eea
 So the facelifting issue exists even when there is no continuous payments. Put differently, the HJB equation \reff{HJB2}-\reff{HJB3} below does not have a classical solution if we use the terminal condition $U(T, \tilde y) = 0$.
\end{rem}

We now write down the HJB equation for the function $U$, even though we already found $U$ without using it  in this case. By \reff{YP-eg2} and \reff{Vtx3}, we have
\bea
\label{HJB2}
\dis \pa_t U + H(t,\pa_{\tilde y} U, \pa_{\tilde y\tilde y} U) =0,
\eea
where the Hamiltonian is:
\bea
\label{HJB3}
\left.\ba{lll}
\dis H:= \sup_{\a, \b, \tilde z} \Big[{1\over 2} \pa_{\tilde y\tilde y} U \tilde z^2 - \pa_{\tilde y} U \big[\a+{1\over 2}(\tilde z+\b)^2 + \tilde F(0)\big] \\
\dis\qq\qq\q + \big[-\a - \k(t)(1-\b)^2 + (\tilde z+\b)(\pa_{\tilde y} U \tilde z + 1-\b)\Big]\\
\dis \q ~= \sup_{\a, \b, \tilde z} \Big[-\a [\pa_{\tilde y} U + 1] + {1\over 2} [\pa_{\tilde y\tilde y} U + \pa_{\tilde y} U]\tilde z^2 \\
\dis \qq\qq\q - [{1\over 2} \pa_{\tilde y} U + 1 ]\b^2 +\k(t)(1-\b)^2- \tilde z \b + \tilde z + \b + \tilde F(0)\Big].
\ea\right.
\eea
Since $\a$ is bounded, we must have $\pa_{\tilde y} U =- 1$, which implies $\pa_{\tilde y\tilde y} U = 0$. Then
\bea
\label{eg-H}
\dis H &=& \sup_{\b, \tilde z} \Big[- {1\over 2} \tilde z^2  - [{1\over 2}  + \k(t)]\b^2 - \tilde z \b + \tilde z + \b + \tilde F(0)\Big] \nonumber\\
\dis &=& \sup_{\b, \tilde z} \Big[- {1\over 2} (\tilde z + \b-1)^2  -  \k(t)(1-\b)^2 + C_0\Big] = C_0,
\eea
with optimal controls $\b^*=1$, $\tilde z^*=0$. Then it is clear that \reff{eg-classical} satisfies \reff{HJB2}-\reff{HJB3} with terminal condition $U(T-, \tilde y) = -\tilde y$.

 \subsection{The principal's problem}
By Proposition \ref{prop-eg-agent}, for any $\tilde y, \a,\b, \tilde Z$ satisfying $\tilde Y_T=0$, we have $V_A(\Xi) = \tilde y + S_0$. Then it is clear that
\beaa
V_P &=& \sup_{S_0\in \dbR, \tilde y \ge R-S_0} V(0, x_0, S_0, \tilde y) = \sup_{S_0\in \dbR, \tilde y \ge R-S_0} \big[ U(0, \tilde y) - S_0\big]\\
&=&\sup_{S_0\in \dbR, \tilde y \ge R-S_0} \big[  C_0 T - \tilde y - S_0\big]  = C_0 T - R.
\eeaa
 From the analysis in the previous subsections, we also have the optimal controls.
 \begin{thm}
 \label{thm-eg-principal}
 Let \reff{model} and Assumption \ref{assum-example} hold. Then the principal's optimal value is:
 \bea
 \label{eg-principal}
 V_P =  C_0 T - R.
\eea
The principal's optimal control is\footnote{$\xi^*$ is arbitrary, since it is not involved in this problem. Since $g_P$ is linear in $m_P$, so the effect of $c^P$ on the continuous part and that on the terminal part cancel each other, and therefore $c^P$ can be arbitrary. Similarly $S_0$ can be arbitrary because it's effect is cancelled by $\a^*$.} :
\bea
\label{eg-contract}
S_0 ~\mbox{arbitrary},\q \a^*_t := {R-S_0\over T} - C_0,\q \b^*_t :=1, \q c^P~\mbox{arbitrary}.
\eea
Moreover, given the above contract, the agent's optimal value and optimal controls are:
\bea
\label{eg-agent} 
V_A(\Xi^*) = R,\q a^*_t = 1,~ c^{A*}_t = {1\over 2}  + h(1).
\eea 
\end{thm}

\begin{rem}
\label{rem-beta}
If we restrict to $\b=0$, as in the standard literature with continuous payments, from \reff{eg-H} we can see that
\beaa
H =  C_0 - \k(t),\q\mbox{with}\q \tilde z^*=1.
\eeaa
Then one can easily see that $U(t_0, \tilde y) = \int_{t_0}^T [C_0-\k(t)]dt - \tilde y$, and $\tilde Z^*_t = 1$ for $t\in [t_0, T]$. However, in this case there is no $\a\in \cA^1_P$ satisfying the constraint $\tilde Y_T=0$. Indeed, given $\b=0$ and $\tilde Z=1$, by \reff{YP-eg2} we have, at $t_0=0$,
\beaa
 \tilde Y_T = \tilde y -\int_0^T  \big[\a_s +{1\over 2} + \tilde F(0) \big]  ds+\int_0^T dB_s = \tilde y - C_0 T- \int_0^T  \a_s ds+ B_T
\eeaa 
That is, to satisfy the constraint, we need to have
\bea
\label{eg-constraint}
\int_0^T  \a_s ds =  \tilde y - C_0 T+ B_T.
\eea
It is well known that there does not exist $\a\in \dbL^1([0, T])$ satisfying the above constraint, not to mention bounded $\a$. So the principal's problem does not  have an optimal contract in this case. This can be viewed as another advantage for allowing $\beta$ in the contract.
\end{rem}


\end{document}